\documentclass[11pt,letterpaper,english]{article}
\usepackage{fullpage}
\usepackage[T1]{fontenc}
\usepackage[utf8]{inputenc}
\usepackage[english]{babel}
\usepackage{tikz,amsmath,amssymb,amsfonts,latexsym,graphicx,amsthm,algorithmicx}
\usepackage{bbm}
\usepackage{todonotes}
\usepackage{thmtools}
\declaretheorem{theorem}

\usepackage{comment}

\usepackage{algorithm}
\usepackage{algpseudocode}
\usepackage{hyperref}
\usepackage[capitalise]{cleveref}

\theoremstyle{plain}
\newtheorem{lemma}[theorem]{Lemma}

\newtheorem{fact}[theorem]{Fact}
\newtheorem{corollary}[theorem]{Corollary}

\newtheorem{remark}[theorem]{Remark}
\theoremstyle{definition}

\newtheorem{definition}[theorem]{Definition}

\newcommand{\OPT}{\mathrm{OPT}}

\newcommand{\TSP}{\mathrm{TSP}}
\newcommand{\ITP}{\mathrm{ITP}}

\newcommand{\opt}{\mathrm{opt}}
\newcommand{\dist}{\mathrm{dist}}
\newcommand{\cost}{\mathrm{cost}}
\newcommand{\rad}{\mathrm{rad}}

\newcommand{\eps}{\epsilon}

\title{Capacitated Vehicle Routing in Graphic Metrics}

\usepackage{authblk}

\author[1]{Tobias M{\"o}mke\thanks{Partially supported by DFG Grant 439522729 (Heisenberg-Grant) and DFG Grant 439637648 (Sachbeihilfe).}}
\author[2]{Hang Zhou\thanks{Partially supported by Hi! PARIS Grant ``Efficiency in Algorithms''.}}
\affil[1]{University of Augsburg, Germany\protect \\ \texttt{moemke@informatik.uni-augsburg.de}}
\affil[2]{École Polytechnique, IP Paris, France\protect \\ \texttt{hzhou@lix.polytechnique.fr}}
\date{}
\begin{document}

\maketitle
\begin{abstract}
We study the \emph{capacitated vehicle routing problem in graphic metrics (graphic CVRP)}. Our main contribution is a new lower bound on the cost of an optimal solution. For graphic metrics, this lower bound is tight  and significantly stronger than the well-known bound for general metrics. The proof of the new lower bound is simple and combinatorial. Using this lower bound, we analyze the approximation ratio of the classical iterated tour partitioning algorithm combined with the TSP algorithms for graphic metrics of Christofides [1976], of M{\"o}mke-Svensson [JACM 2016], and of Seb{\H{o}}-Vygen [Combinatorica 2014]. In particular, we obtain a 1.95-approximation for the graphic CVRP.
\end{abstract}
\newpage{}

\section{Introduction}

Given a metric space with  a set of $n$ \emph{terminals}, a \emph{depot}, and an integer \emph{tour capacity} $k$, the \emph{capacitated vehicle routing problem (CVRP)} asks for a minimum length collection of tours starting and ending at the depot such that the number of terminals covered by each tour is at most $k$ and those tours together cover all terminals.

The CVRP was introduced by Dantzig and Ramser in 1959~\cite{dantzig1959truck}.
It is a generalization of the \emph{traveling salesman problem (TSP)} and is one of the most studied problems in Operations Research.
Books have been dedicated to vehicle routing problems, e.g.,\ \cite{toth2002vehicle,golden2008vehicle,crainic2012fleet,anbuudayasankar2016models}.
Yet, these problems remain challenging, both from a practical and a theoretical perspective.

The most popular polynomial-time approximation for the CVRP is a \emph{simple} algorithm, called \emph{iterated tour partitioning (ITP)}.
The ITP algorithm was introduced in 1985 by Haimovich and Rinnooy~Kan~\cite{haimovich1985bounds}.
Altinkemer and Gavish~\cite{altinkemer1990heuristics} showed that, in general metric spaces, the approximation ratio of the ITP algorithm  is at most $1+\left(1-\frac{1}{k}\right)\alpha$, where $\alpha\geq 1$ is the approximation ratio of a TSP algorithm.
Bompadre, Dror, and Orlin~\cite{bompadre2006improved} improved this bound to $1+\left(1-\frac{1}{k}\right)\alpha-\Omega\left(\frac{1}{k^3}\right)$.
The ratio for the CVRP in general metric spaces was recently improved by Blauth, Traub, and Vygen~\cite{blauth2021improving} to $1+\alpha-\eps$, where $\eps$ is at least $\frac{1}{3000}$.
Additionally, the best-to-date approximation ratio $\alpha$ for the TSP in general metrics is $1.5 - 10^{-36}$ by Karlin, Klein, and Oveis Gharan~\cite{karlin2021slightly}, improving upon the ratio of $1.5$ from Christofides~\cite{christofides1976worst} and Serdyukov~\cite{Ser78,BVS20}.
Consequently, the best-to-date approximation ratio for the metric CVRP stands at roughly $2.5-10^{-36}-\frac{1}{3000}$.

In this work, we focus on \emph{graphic} metrics, where the distance between two vertices is the length (i.e., number of edges) of a shortest path in a given unweighted graph.\footnote{Technically, each edge in the given graph has a cost of one and forming a shortest path metric introduces new edges of larger costs. 
    For TSP and CVRP, however, we may assume without loss of generality that such edges are replaced by a shortest path of cost-one edges.
When referring to an edge, we therefore always refer to an edge in the \emph{original} graph.
For further details we refer to the excellent survey of Vygen~\cite{Vyg12}.
}
The graphic TSP has attracted much attention. 
It captures the difficulty of the metric TSP in the sense that it is \textsf{APX} hard~\cite{GKP95} and the lower bound $\frac{4}{3}$ on the integrality gap of the Held-Karp relaxation for the metric TSP is established using an instance of the graphic TSP, see, e.g.,~\cite{Vyg12}.
Oveis Gharan, Saberi, and Singh~\cite{gharan2011randomized} gave the first approximation algorithm with ratio strictly better than 1.5 for the graphic TSP.
M{\"o}mke and Svensson~\cite{MS16} obtained a 1.461-approximation for the graphic TSP by introducing the powerful idea of \emph{removable pairings}.
Mucha~\cite{mucha2014frac} improved the analysis of~\cite{MS16}
to achieve an approximation ratio of $\frac{13}{9}\approx 1.444$.
The best-to-date approximation for the graphic TSP is due to Seb{\H{o}} and Vygen~\cite{sebHo2014shorter}, who gave an elegant 1.4-approximation algorithm using a special kind of ear-decomposition.

We study the capacitated vehicle routing problem in \emph{graphic} metrics, called the \emph{graphic CVRP}.
The graphic CVRP is a generalization of the graphic TSP. 
Combining the graphic TSP algorithm of Seb{\H{o}} and Vygen~\cite{sebHo2014shorter} and the algorithm of Blauth, Traub, and Vygen~\cite{blauth2021improving} for the CVRP in general metrics, the best-to-date approximation ratio for the graphic CVRP stands at roughly $2.4-\frac{1}{3000}$.
In this work, we reduce the approximation ratio for the graphic CVRP to 1.95.
\begin{center}
\renewcommand{\arraystretch}{1.4}
\begin{tabular}{ |c|l|l| } 
 \hline
  & \emph{general} metrics & \emph{graphic} metrics \\\hline
 TSP & $1.5-10^{-36}$~\cite{karlin2021slightly} & 1.4~\cite{sebHo2014shorter} \\\hline
CVRP & $2.5-10^{-36}-\frac{1}{3000}$~\cite{blauth2021improving} & {\bf 1.95~[this work]}\\\hline
\end{tabular}
\end{center}
\vspace{1mm}
\subsection{Our Results}

Our main result depends on the standard definition of the \emph{radius cost} introduced by Haimovich and Rinnooy Kan~\cite{haimovich1985bounds}.

\begin{definition}[radius cost, \cite{haimovich1985bounds}]
Let $V$ be a set of terminals and let $O$ be the depot.
For every vertex $v\in V$, let $\dist(v)$ denote the $v$-to-$O$ distance in the graph.
Let $\rad$ denote the \emph{radius cost}, defined by $\rad=\frac{2}{k}\sum_{v\in V} \dist(v).$
\end{definition}

Let $\opt$ denote the cost of an optimal solution to the graphic CVRP. 
It is well-known that $\opt\geq \rad$ in general metrics~\cite{haimovich1985bounds}.

Our main technical contribution is a new lower bound on $\opt$ in graphic metrics, stated in the following Structure Theorem (\cref{thm:structure}).
The new lower bound is tight in graphic metrics, see \cref{fig:tight}.
It is significantly stronger than the known lower bound stated above.
The proof of the new lower bound, in \cref{sec:structure}, is simple and combinatorial.

\begin{theorem}[Structure Theorem]
\label{thm:structure}
Consider the graphic CVRP with a set $V$ of $n$ terminals, a depot $O$, and a tour capacity $k$.
We have
\[
    \opt\geq \rad + \frac{n}{2}-\frac{n}{2k^2}.
\]
\end{theorem}

The Structure Theorem implies easily the following theorem on the approximation ratio of the classical iterated tour partitioning algorithm for the graphic CVRP.
The proof of \cref{thm:CVRP} is in \cref{sec:proof-CVRP}.

\begin{theorem}
\label{thm:CVRP}
Consider the graphic CVRP with a set $V$ of $n$ terminals, a depot $O$, and an arbitrary tour capacity $k$.
Let $S$ be a  traveling salesman tour on $V\cup \{O\}$ of cost at most $\beta\cdot n+\gamma\cdot \opt_{\TSP}$ with $\beta\geq \frac{1}{2}$ and $\gamma\geq 0$, where $\opt_{\TSP}$ denotes the cost of an optimal TSP tour.
Then the iterated tour partitioning algorithm applied on $S$ yields a $\left(\beta+\gamma+\frac{1}{2}\right)$-approximate solution to the graphic CVRP.
\end{theorem}

As a consequence of \cref{thm:CVRP}, in \cref{cor:all} we bound the approximation ratio of the iterated tour partitioning algorithm combined with the TSP algorithms for graphic metrics of Christofides~\cite{christofides1976worst}, of M{\"o}mke-Svensson~\cite{MS16}, and of Seb{\H{o}}-Vygen~\cite{sebHo2014shorter}.

\begin{figure}[ht]
\centering
\includegraphics[scale=0.45]{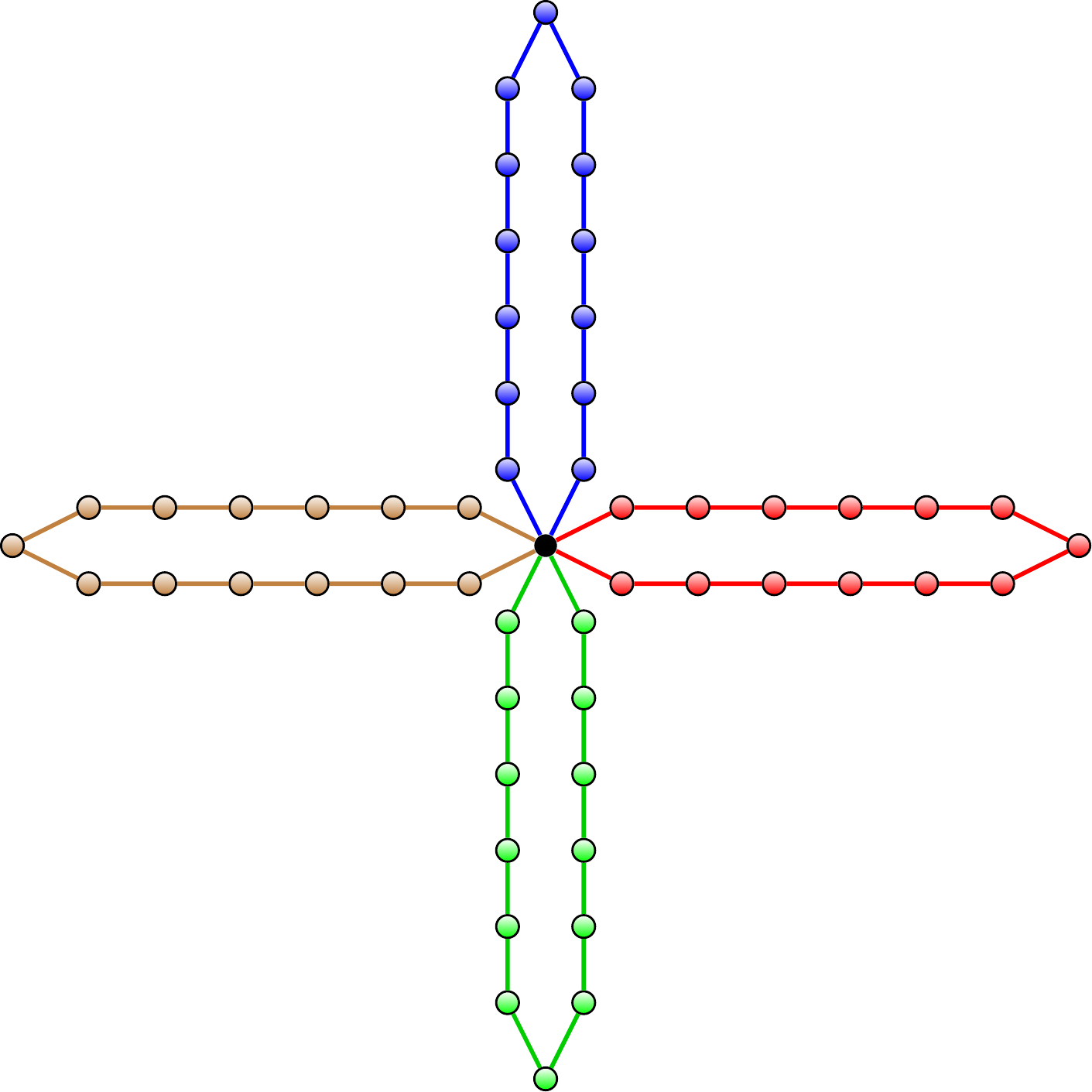}
\caption{
A tight instance for the lower bound in the Structure Theorem (\cref{thm:structure}).
Let $k$ be an odd integer. Let $n$ be a multiple of $k$.
The graph consists of $\frac{n}{k}$ cycles going through the depot (the central black node), each visiting exactly $k$ terminals. (In the example, $k=13$ and $n=52$, and the graph consists of 4 cycles.)
Each cycle consists of $k+1$ edges, so the cost $\opt$ of an optimal solution equals $(k+1)\cdot \frac{n}{k}=n+\frac{n}{k}$.
At the same time, a simple calculation gives  $\rad=\frac{n}{2}+\frac{n}{k}+\frac{n}{2k^2}$. Thus $\rad+\frac{n}{2}-\frac{n}{2k^2}=n+\frac{n}{k}$.
Hence the lower bound in the Structure Theorem is tight.
}
\label{fig:tight}
\end{figure}

\begin{corollary}
\label{cor:all}
Consider the graphic CVRP with a set $V$ of terminals, a depot $O$, and an arbitrary tour capacity $k$.
We have:
\begin{enumerate}
    \item Let $S_1$ be a traveling salesman tour on $V\cup \{O\}$ computed by the Christofides algorithm~\cite{christofides1976worst}.
    The iterated tour partitioning algorithm applied on $S_1$ yields a 2-approximate solution.
    \item Let $S_2$ be a traveling salesman tour on $V\cup \{O\}$ that is the better one of the two solutions computed by the M{\"o}mke-Svensson algorithm~\cite{MS16} and the Christofides algorithm.
    The iterated tour partitioning algorithm applied on $S_2$ yields a $\left(2-\frac{1}{24}\right)$-approximate solution.
    \item Let $S_3$ be a traveling salesman tour on $V\cup \{O\}$ that is the better one of the two solutions computed by the Seb{\H{o}}-Vygen algorithm~\cite{sebHo2014shorter} and the Christofides algorithm.
    The iterated tour partitioning algorithm applied on $S_3$ yields a $1.95$-approximate solution.
\end{enumerate}
\end{corollary}

To prove \cref{cor:all}, the requirement $\beta\geq \frac{1}{2}$ is crucial in the application of \cref{thm:CVRP}.
Intuitively, the threshold $\frac{1}{2}$ comes from the additive term $\frac{n}{2}$ in the Structure Theorem.
For the Christofides algorithm~\cite{christofides1976worst}, the requirement $\beta\geq \frac{1}{2}$  follows immediately since in a graph with $n+1$ vertices the algorithm combines a spanning tree of cost exactly $n$ with a matching of cost at most  $\frac{1}{2}\cdot \opt_{\TSP}$.
For the M{\"o}mke-Svensson algorithm~\cite{MS16} we may assume $\beta = \frac{1}{3}$, see \cref{lem:MS16};
and for the Seb{\H{o}}-Vygen algorithm~\cite{sebHo2014shorter} we may assume $\beta = 0$, see \cref{lem:Sebo-Vygen}.
Comparing their solutions with the solution of the Christofides algorithm and returning the better of the solutions lead to $\beta\geq \frac{1}{2}$.
The proof of \cref{cor:all} is in \cref{sec:proof-cor}.

\begin{remark}
In \cref{cor:all}, we analyze the Christofides algorithm but not the algorithm of Karlin, Klein, and Oveis Gharan~\cite{karlin2021slightly}, because of the simplicity of the Christofides algorithm, and also because the algorithm in \cite{karlin2021slightly} requires computing a random spanning tree whose cost might be greater than $n$ due to multiple edges.

Also note that the algorithm of Blauth, Traub, and Vygen~\cite{blauth2021improving} does not seem to yield an improved approximation for the graphic CVRP, because worst-case instances of the ITP algorithm for graphic metrics are different from those for general metrics.
\end{remark}

\paragraph{Open Questions.}
In our work, we obtain an approximation for the graphic CVRP of ratio $1.95$.
To that end, we combine the Seb{\H{o}}-Vygen algorithm~\cite{sebHo2014shorter}, the Christofides algorithm~\cite{christofides1976worst}, and the iterated tour partitioning algorithm.
The main open question to improve upon the 1.95 factor  for the graphic CVRP.
In particular, it is an interesting open question to analyze the cost of the solution computed by the Seb{\H{o}}-Vygen algorithm  in terms of the number of vertices, which might lead to an improvement over the 1.95 factor.

\subsection{Related Work}
\label{sec:related}

\paragraph{Graphic TSP.}
Besides the classical version of the graphic TSP mentioned previously, the \emph{$s$-$t$-path} version of the graphic TSP has also been well-studied~\cite{sebHo2014shorter,MS16,mucha2014frac,AKS15,Gao13,TV18,TVZ20}, with the best-to-date approximation ratio being $1.4+\eps$ due to Traub, Vygen, and Zenklusen~\cite{TVZ20}.

Special cases of the graphs have been studied for the graphic TSP, e.g.,
cubic graphs~\cite{GLS05,AGG18,BSSS14,MS16,CLS15,DL18,DKM17},
cubic bipartite graphs~\cite{KR16,Zuy18},
graphs of degree at most $3$ and claw-free graphs~\cite{MS16},
and graphs of degree at most $4$~\cite{New20}.

\paragraph{CVRP in Other Metrics.}
The CVRP has been extensively studied in other metrics: trees and bounded treewidth~\cite{MZ22unit,jayaprakash2021approximation,becker2019framework,becker2018tight,asano2001new}, Euclidean~\cite{haimovich1985bounds,asano1997covering,adamaszek2010ptas,das2015quasipolynomial,jayaprakash2021approximation}, planar and bounded-genus graphs~\cite{becker2017quasi,becker2019ptas,cohen2020light}, 
graphs of bounded highway dimension~\cite{becker2018polynomial}, and minor-free graphs~\cite{cohen2020light}.


\paragraph{CVRP with Arbitrary Unsplittable Demands}
A natural way to generalize the unit demand version of the CVRP is to allow terminals to have arbitrary unsplittable demands, which is called the ``unsplittable'' version of the CVRP.
The first constant-factor approximation algorithm for the unsplittable CVRP in general metrics is due to Altinkemer and  Gavish~\cite{altinkemer1987heuristics} and has a ratio of $2+\alpha$, where $\alpha$ is the approximation ratio of a TSP algorithm. The approximation ratio for the unsplittable CVRP was only recently improved to $2+\alpha-2\eps$ by Blauth, Traub, and Vygen~\cite{blauth2021improving}, where $\epsilon$ is at least $\frac{1}{3000}$.
Very recently, Friggstad et al.~\cite{friggstad2021improved} further improved the ratio to roughly $\ln 2 + \alpha+1+\epsilon$ for any constant $\epsilon>0$.
This problem has also been studied on trees~\cite{MZ22unsplit} and in the Euclidean space~\cite{GMZ22}.

\paragraph{Iterated Tour Partitioning.}
The \emph{iterated tour partitioning (ITP)} algorithm is the most popular polynomial-time approximation for the CVRP, and is  very simple. 
This algorithm first computes a traveling salesman tour (ignoring the capacity constraint) using some other algorithm as a black box, then partitions the tour into segments such that the number of terminals in each segment is at most $k$, and finally, for each segment, connects the endpoints of that segment to the depot so as to make a tour.
The ITP algorithm was introduced and refined by Haimovich and Rinnooy~Kan~\cite{haimovich1985bounds} and Altinkemer and Gavish~\cite{altinkemer1990heuristics} in the 1980s.

The approximation ratio of the ITP algorithm has been well-studied, and bounds on the ratio have been utilized in the design of approximation algorithms, see, e.g.,~\cite{bompadre2006improved}.
Li and Simchi-Levi~\cite{li1990worst} showed that the ITP algorithm cannot lead to a better-than $(2-\frac{1}{k})$-approximation in general metrics.
Blauth, Traub, and Vygen~\cite{blauth2021improving} exploited properties of tight instances for the ITP algorithm, and used those properties to design the best-to-date approximation algorithm for the metric CVRP.
The performance of the ITP algorithm has also been studied in the special case when the terminals are uniform random points in the Euclidean plane~\cite{bompadre2007probabilistic,MZ21}.

Because of its simplicity, the ITP algorithm is versatile and has been adapted to vehicle routing problems in other settings, e.g., with pick-up and delivery services~\cite{mosheiov1998vehicle}, or under the constraints on the total distance traveled by each vehicle~\cite{li1992distance}.

\section{Preliminaries}

In this section, we introduce some notations and we formally define the problem. 

Let $G=(V\cup \{O\},E)$ be a connected, unweighted, and undirected graph, where $V$ is a set of \emph{terminals}, vertex $O$ is the \emph{depot}, and $E$ is the set of edges.
Let $n$ denote the number of terminals in $V$.
Let $V(G)$ denote $V\cup\{O\}$.
For each $v\in V$, let $\dist(v)$ denote the number of edges on a $v$-to-$O$ shortest path in $G$.
A \emph{tour} $T$ in $G$ is a path $z_1 z_2 \dots z_p$ for some $p\in \mathbb{N}$ such that $z_i\in V(G)$ for each $i\in[1,p]$; $z_1=z_p=O$; and $(z_i,z_{i+1})\in E$ for each $i\in [1,p-1]$.
The \emph{cost} of the tour $T$, denoted by $\cost(T)$, is the number of edges on that tour, i.e., $\cost(T)=p-1$.
Each terminal has \emph{unit demand}, which must be covered by a single tour.
Let $k\in[1,n]$ be an integer \emph{tour capacity}, i.e., each tour can cover the demand of at most $k$ terminals.
\begin{definition}[graphic CVRP]
An instance of the \emph{capacitated vehicle routing problem in graphic metrics (graphic CVRP)} consists of
\vspace{-0.2em}
\begin{itemize}
\itemsep-0.2em 
    \item a connected, unweighted, and undirected graph $G=(V\cup\{O\},E)$ where $V$ is a set of $n$ \emph{terminals} and $O$ is the \emph{depot};
    \item a positive integer \emph{tour capacity} $k$ such that $k\in[1,n]$.
\end{itemize}
\vspace{-0.1em}
A feasible solution is a set of tours such that
\vspace{-0.2em}
\begin{itemize}
\itemsep-0.2em 
    \item each tour starts and ends at $O$,
    \item each tour covers the demand of at most $k$ terminals,
    \item the demand of each terminal is covered by one tour.
\end{itemize}
The goal is to find a feasible solution such that the total cost of the tours is minimum.

\end{definition}
Let $\OPT$ denote an optimal solution to the graphic CVRP, and let $\opt$ denote the cost of $\OPT$.

\section{Proof of the Structure Theorem (Theorem~\ref{thm:structure})}
\label{sec:structure}
In this section, we prove the Structure Theorem (\cref{thm:structure}).
The key in the analysis is the following Structure Lemma.

\begin{lemma}[Structure Lemma]
\label{lem:structure}
Let $T := z_1 z_2\dots z_p$ for some $p\in \mathbb{N}$ be a tour in $\OPT$.
Let $U\subseteq V$ denote the set of terminals whose demands are covered by $T$.
Let $D$ denote $\sum_{v\in U} \dist(v)$.
Then
\begin{equation*}
    \cost(T)\geq
    \frac{2D}{|U|}+
     \frac{|U|}{2}-\frac{1}{2|U|}.
\end{equation*}
\end{lemma}

In \cref{sec:struc-lemma-struc-thm}, we show how the Structure Lemma implies the Structure Theorem, and in \cref{sec:proof-lem:structure}, we prove the Structure Lemma.

\subsection{Using the Structure Lemma to prove the Structure Theorem}
\label{sec:struc-lemma-struc-thm}
Let $T_1,\dots,T_\ell$ denote the tours in $\OPT$ for some $\ell\geq 1$.
For each $i\in[1,\ell]$, let $U_i\subseteq V$ denote the set of terminals whose demands are covered by $T_i$.
Let $D_i$ denote $\sum_{v\in U_i} \dist(v)$.
We apply the Structure Lemma (\cref{lem:structure}) on $T_i$ and obtain
\begin{equation*}
    \cost(T_i)\geq
      \frac{2D_i}{|U_i|}+
     \frac{|U_i|}{2}-\frac{1}{2|U_i|}.
\end{equation*}
Summing over all tours $T_i$ and noting that $\displaystyle\sum_{i=1}^{\ell}|U_i|=n$, we have \begin{equation*}
\opt=\sum_{i=1}^\ell\cost(T_i)\geq \frac{n}{2}+\sum_{i=1}^\ell\left(\frac{2D_i}{|U_i|}-\frac{1}{2|U_i|}\right).
\end{equation*}
Define 
\[
Z:=\sum_{i=1}^\ell\left(\frac{2D_i}{|U_i|}-\frac{1}{2|U_i|}\right)- \rad+\frac{n}{2k^2}.\]

To show the claim in the Structure Theorem, it suffices to show that $Z$ is non-negative.
Using $\displaystyle\rad=\frac{\sum_{i=1}^{\ell} 2D_i}{k}$ and $\displaystyle n=\sum_{i=1}^{\ell} |U_i|$, we have
\[Z=\sum_{i=1}^{\ell} \left(\frac{2D_i}{|U_i|}-\frac{1}{2|U_i|}-\frac{2D_i}{k}+\frac{|U_i|}{2k^2}\right)=\sum_{i=1}^\ell \frac{2k(k-|U_i|)(2D_i-1)+(k-|U_i|)^2}{2k^2|U_i|}\geq 0,\]
where the inequality follows from $|U_i|\leq k$ and $D_i\geq 1$ (since $\dist(v)\geq 1$ for each $v\in U_i$).

This completes the proof of the Structure Theorem (\cref{thm:structure}).

\subsection{Proof of the Structure Lemma (Lemma~\ref{lem:structure})}
\label{sec:proof-lem:structure}

Let $W$ be a multi-set of terminals defined by \[W := \{z_i\mid i\in[1,p] \text{ and } z_i\neq O\}.\]
Let $R := \max_{v\in W} \dist(v)$.
Let $\Delta := R-\frac{D}{|U|}$.
It is easy to see that $\Delta\geq 0$.

If $\Delta\geq \frac{|U|}{4}$, we have  \[\displaystyle\cost(T)\geq 2R\geq \frac{|U|}{2}+\frac{2D}{|U|},\] the claim follows.

In the rest of the proof, we assume that $\Delta<\frac{|U|}{4}$.
Let $j\in[1,p]$ be such that  $\dist(z_j)=R$, breaking ties arbitrarily.
We split the tour $T$ at vertex $z_j$, obtaining two paths $P_1:=z_1 z_2\dots z_j$ and $P_2:=z_j z_{j+1}\dots z_p$.
For each $i\in[1,R-1]$, let $x_i$ (resp.\ $y_i$) be a vertex on $P_1$ (resp.\ $P_2$), such that $\dist(x_i)$ (resp.\ $\dist(y_i)$) equals $i$; breaking ties arbitrarily.
Let $A$ denote the multi-set of vertices $\{z_j,x_1,\dots, x_{R-1},y_1,\dots,y_{R-1}\}$.
Observe that $A\subseteq W$.
Define another multi-set $B:=W\setminus A$.
Let $U_A\subseteq U$ denote the set of vertices $u\in U$ such that $u$ has at least one occurrence in $A$.
Let $U_B := U\setminus U_A$.
Observe that each element in $U_B$ has at least one occurrence in $B$, so \[|B|\geq |U_B|=|U|-|U_A|.\]
From the definition of $W$ and since $O\notin W$, we have
\[\cost(T)\geq |W|+1=|A|+|B|+1=2\cdot R+|B|.\]
Noting that $R=\frac{D}{|U|}+\Delta$, we have
\begin{equation}
\label{eqn:cost-T-1}
\cost(T)\geq \frac{2D}{|U|} +2\Delta+|U|-|U_A|.
\end{equation}

To lower bound $\cost(T)$, we upper bound $|U_A|$ in the following lemma, whose proof is elementary.
\begin{lemma}
\label{lem:U_A}
$|U_A|\leq \sqrt{4|U|\cdot \Delta+1}$.
\end{lemma}

\begin{proof}
First, consider the case when $|U_A|$ is even. Suppose $|U_A|=2m$ for some $m\in \mathbb{N}$.
Since $U_A\subseteq A$ and using the definition of $A$, we have
\[\sum_{v\in U_A}\dist(v)\leq R+ 2(R-1)+2(R-2)+\cdots+2(R-m+1)+(R-m)=2m\cdot R-m^2.\]
Since $\dist(v)\leq R$ for each $v\in U_B\subseteq W$, we have
\[\sum_{v\in U_B}\dist(v)\leq |U_B|\cdot R=(|U|-|U_A|)\cdot R= |U|\cdot R-2m\cdot R.\]
Thus \[D=\sum_{v\in U}\dist(v)\leq |U|\cdot R-m^2.\]
Since $D=|U|\cdot (R-\Delta)$, we have \[m\leq \sqrt{|U|\cdot \Delta}.\] 
The claim follows since $|U_A|=2m$.

Next, consider the case when $|U_A|$ is odd.
Suppose $|U_A|=2m+1$ for some $m\in \mathbb{N}$.
Since $U_A\subseteq A$ and using the definition of $A$, we have
\[\sum_{v\in U_A}\dist(v)\leq R+2(R-1)+2(R-2)+\dots + 2(R-m)= (2m+1)\cdot R-m(m+1).\]
Since $\dist(v)\leq R$ for each $v\in U_B\subseteq W$, we have
\[\sum_{v\in U_B}\dist(v)\leq |U_B|\cdot R= (|U|-|U_A|)\cdot R=|U|\cdot R-(2m+1)\cdot R.\]
Thus \[D=\sum_{v\in U}\dist(v)\leq |U|\cdot R-m(m+1).\]
Since $D=|U|\cdot (R-\Delta)$, we have \[m\leq \frac{\sqrt{4|U|\cdot \Delta+1}-1}{2}.\] 
The claim follows since $|U_A|=2m+1$.
\end{proof}

From \eqref{eqn:cost-T-1} and \cref{lem:U_A}, we have
\begin{equation}
\label{eqn:cost-T-2}
    \cost(T)\geq \frac{2D}{|U|}+2\Delta+|U|-\sqrt{4|U|\cdot \Delta+1}.
\end{equation}
To lower bound $\cost(T)$, we use the following simple fact. 
\begin{fact}
\label{fact:Delta}
\begin{equation*}
     2\Delta-\sqrt{4|U|\cdot \Delta+1}\geq -\frac{|U|}{2}-\frac{1}{2|U|}.
\end{equation*}
\end{fact}
\begin{proof}
Define \[Y:=2\Delta-\sqrt{4|U|\cdot \Delta+1}+\frac{|U|}{2}+\frac{1}{2|U|}.\]
It suffices to show that $Y\geq 0$.
Let $t:=\sqrt{4|U|\cdot \Delta+1}$.
Thus $\displaystyle\Delta=\frac{t^2-1}{4|U|}$.
We have \[Y=2\cdot\frac{t^2-1}{4|U|}-t+\frac{|U|}{2}+\frac{1}{2|U|}=\left(\frac{t}{\sqrt{2|U|}}-\sqrt{\frac{|U|}{2}}\right)^2\geq 0.\]
\end{proof}

From \eqref{eqn:cost-T-2} and \cref{fact:Delta}, we conclude that
\[\cost(T)\geq \frac{2D}{|U|}+\frac{|U|}{2}-\frac{1}{2|U|}.\]
This completes the proof of the Structure Lemma (\cref{lem:structure}).


\section{Proof of Theorem~\ref{thm:CVRP}}
\label{sec:proof-CVRP}
Let $S$ be a traveling salesman tour on $V\cup \{O\}$.
Let $\ITP(S)$ denote the cost of the solution computed by the iterated tour partitioning algorithm applied on $S$. 
Altinkemer and Gavish~\cite{altinkemer1990heuristics} showed the following bound. 

\begin{lemma}[\cite{altinkemer1990heuristics}]
\label{lem:ITP-upper-bound}
$\ITP(S)\leq \rad+\left(1-\frac{1}{k}\right)\cdot\cost(S).$
\end{lemma}
From \cref{lem:ITP-upper-bound} and the assumption on $\cost(S)$, we have 
\begin{align*}
    \ITP(S)&\leq \rad +\left(1-\frac{1}{k}\right)\cdot\left(\beta\cdot n+\gamma\cdot \opt_{\TSP}\right)\\
    &= \left(\rad+\frac{n}{2}-\frac{n}{2k^2} \right) + \left(\beta-\frac{1}{2}+\frac{1}{2k^2}-\frac{\beta}{k}\right)\cdot n+\left(1-\frac{1}{k}\right)\cdot\gamma\cdot\opt_{\TSP}.
\end{align*}
From the Structure Theorem (\cref{thm:structure}), 
\[\rad+\frac{n}{2}-\frac{n}{2k^2}\leq \opt.\]
Since $\beta\geq \frac{1}{2}$ and $k\geq 1$, we have \[\frac{1}{2k^2}-\frac{\beta}{k}\leq 0.\]
Combining, we have
\[\ITP(S)\leq \opt + \left(\beta-\frac{1}{2}\right)\cdot n+\gamma\cdot\opt_{\TSP}\leq\left(\beta+\gamma+\frac{1}{2}\right)\cdot \opt,\]
where the last inequality follows from $n\leq \opt$ and $\opt_{\TSP}\leq \opt$.

This completes the proof of \cref{thm:CVRP}.

\section{Proof of Corollary~\ref{cor:all}}

\label{sec:proof-cor}
In this section, we analyze the approximation ratios for the graphic CVRP using known algorithms for the graphic TSP.

\begin{lemma}[adaptation from Christofides~\cite{christofides1976worst}]
\label{lem:Christofides}
Given a connected graph $G$ with $n+1$ vertices, the Christofides algorithm computes a traveling salesman tour of cost at most $n + \frac{1}{2}\cdot\opt_{\TSP}$.
\end{lemma}

\begin{proof}
Christofides algorithm~\cite{christofides1976worst} computes a minimum spanning tree $T$ and adds a minimum-cost perfect matching $M$ of the odd degree vertices as parity correction.
Since $G$ is connected, the cost of $T$ is exactly $n$.
By~\cite{christofides1976worst}, the cost of $M$ is at most $\frac{1}{2}\cdot\opt_{\TSP}$.
The claim follows.
\end{proof}

The first claim in \cref{cor:all} follows from \cref{thm:CVRP,lem:Christofides}.

\begin{lemma}[adaptation from Mucha~\cite{mucha2014frac} and Mömke and Svensson~\cite{MS16}]
\label{lem:MS16}
Given a connected graph $G$ with $n+1$ vertices, the Mömke-Svensson algorithm computes a traveling salesman tour of cost at most
$\frac{n}{3}+\frac{10}{9}\cdot\opt_{\TSP}$.
\end{lemma}

\begin{proof}   
Letting $n'=n+1$, it suffices to show that, given a connected graph $G$ with $n'$ vertices, the Mömke-Svensson algorithm computes a traveling salesman tour of cost at most $\frac{n'}{3}+\frac{10}{9}\cdot\opt_{\TSP}-\frac{2}{3}$.

The proof is by induction.
Let $X$ denote the cost of the traveling salesman tour computed by the Mömke-Svensson algorithm.

First, consider the case when $G$ is $2$-vertex-connected.
The Mömke-Svensson algorithm computes some \emph{circulation} of cost $c$ and they show in Lemma~4.2 of \cite{MS16}:
\[X\leq \frac{4}{3}\cdot n'+ \frac{2}{3}\cdot c - \frac{2}{3}.\]
Mucha shows in Corollary~1 of \cite{mucha2014frac}:
\[c\leq \frac{5}{3}\cdot\text{opt}_{\TSP} - \frac{3}{2}\cdot n'.\] The claim follows.
    
Next, consider the case when $G$ is not $2$-vertex-connected. 
There is a vertex $v$ such that removing $v$ disconnects the graph.
We therefore identify two subgraphs $G_1 = (V_1,E_1)$ and $G_2 = (V_2,E_2)$
such that $V_1 \cup V_2 = V(G)$, $V_1 \cap V_2 = \{v\}$ and $E = E_1 \cup E_2$.
Observe that $|V_1|+|V_2|=n'+1.$
From the induction, the claim holds for both $G_1$ and $G_2$.
Let $\text{opt}_1$ (resp.\ $\text{opt}_2$) be the optimal cost of a TSP solution in $G_1$ (resp.\ $G_2$).
Then $\opt_{\TSP}=\opt_1+\opt_2$.
We have
\[X \leq  \left(\frac{|V_1|}{3}+\frac{10}{9}\cdot\text{opt}_1-\frac{2}{3}\right) + \left(\frac{|V_2|}{3}+\frac{10}{9}\cdot\text{opt}_2-\frac{2}{3}\right) 
= \frac{n'+1}{3}+ \frac{10}{9}\cdot \opt_{\TSP}-\frac{4}{3},\]
which yields the claim.
\end{proof}

From \cref{lem:Christofides,lem:MS16}, and since $S_2$ is the better one of the two solutions computed by the Christofides algorithm and by the Mömke-Svensson algorithm, we have 
\begin{equation}
\label{eqn:S2}
\cost(S_2)\leq \frac{1}{4}\cdot\left(n+\frac{1}{2}\cdot\opt_{\TSP}\right)+\frac{3}{4}\cdot\left(\frac{n}{3}+\frac{10}{9}\cdot\opt_{\TSP}\right)=\frac{1}{2}\cdot n+\frac{23}{24}\cdot\opt_{\TSP}.
\end{equation}

The second claim in \cref{cor:all} follows from \cref{thm:CVRP} and \eqref{eqn:S2}.

\begin{lemma}[Seb\H{o} and Vygen~\cite{sebHo2014shorter}]
\label{lem:Sebo-Vygen}
Given a connected graph $G$ with $n+1$ vertices, the Seb\H{o}-Vygen algorithm computes a traveling salesman tour  of cost at most $\frac{7}{5}\cdot\opt_{\TSP}$. 
\end{lemma}

From \cref{lem:Christofides,lem:Sebo-Vygen}, and since $S_3$ is the better one of the two solutions computed by the Christofides algorithm and by the Seb\H{o}-Vygen algorithm, we have 
\begin{equation}
\label{eqn:S3}
\cost(S_3)\leq \frac{1}{2}\cdot\left(n+\frac{1}{2}\cdot\opt_{\TSP}\right)+\frac{1}{2}\cdot\left(\frac{7}{5}\cdot\opt_{\TSP}\right)=\frac{1}{2}\cdot n+\frac{19}{20}\cdot\opt_{\TSP}.
\end{equation}

The last claim in \cref{cor:all} follows from \cref{thm:CVRP} and \eqref{eqn:S3}.

\paragraph{Acknowledgments.}
We thank Zipei Nie for helpful discussions in mathematics.
\bibliographystyle{alphaabbr}
\bibliography{references}

\begin{thebibliography}{BSvdSS14}

\bibitem[ACL10]{adamaszek2010ptas}
A.~Adamaszek, A.~Czumaj, and A.~Lingas.
\newblock {PTAS for $k$-tour cover problem on the plane for moderately large
  values of $k$}.
\newblock {\em International Journal of Foundations of Computer Science},
  21(06):893--904, 2010.

\bibitem[AG87]{altinkemer1987heuristics}
K.~Altinkemer and B.~Gavish.
\newblock Heuristics for unequal weight delivery problems with a fixed error
  guarantee.
\newblock {\em Operations Research Letters}, 6(4):149--158, 1987.

\bibitem[AG90]{altinkemer1990heuristics}
K.~Altinkemer and B.~Gavish.
\newblock Heuristics for delivery problems with constant error guarantees.
\newblock {\em Transportation Science}, 24(4):294--297, 1990.

\bibitem[AGG18]{AGG18}
N.~Agarwal, N.~Garg, and S.~Gupta.
\newblock A 4/3-approximation for {TSP} on cubic 3-edge-connected graphs.
\newblock {\em Oper. Res. Lett.}, 46(4):393--396, 2018.

\bibitem[AGM16]{anbuudayasankar2016models}
S.~P. Anbuudayasankar, K.~Ganesh, and S.~Mohapatra.
\newblock {\em Models for practical routing problems in logistics}.
\newblock Springer, 2016.

\bibitem[AKK01]{asano2001new}
T.~Asano, N.~Katoh, and K.~Kawashima.
\newblock A new approximation algorithm for the capacitated vehicle routing
  problem on a tree.
\newblock {\em Journal of Combinatorial Optimization}, 5(2):213--231, 2001.

\bibitem[AKS15]{AKS15}
H.~An, R.~D. Kleinberg, and D.~B. Shmoys.
\newblock Improving Christofides' algorithm for the $s$-$t$ path {TSP}.
\newblock {\em Journal of the ACM (JACM)}, 62(5):34:1--34:28, 2015.

\bibitem[AKTT97]{asano1997covering}
T.~Asano, N.~Katoh, H.~Tamaki, and T.~Tokuyama.
\newblock Covering points in the plane by $k$-tours: towards a polynomial time
  approximation scheme for general $k$.
\newblock In {\em ACM Symposium on Theory of Computing (STOC)}, pages 275--283,
  1997.

\bibitem[BDO06]{bompadre2006improved}
A.~Bompadre, M.~Dror, and J.~B. Orlin.
\newblock Improved bounds for vehicle routing solutions.
\newblock {\em Discrete Optimization}, 3(4):299--316, 2006.

\bibitem[BDO07]{bompadre2007probabilistic}
A.~Bompadre, M.~Dror, and J.~B. Orlin.
\newblock Probabilistic analysis of unit-demand vehicle routeing problems.
\newblock {\em Journal of Applied Probability}, 44(1):259--278, 2007.

\bibitem[Bec18]{becker2018tight}
A.~Becker.
\newblock {A tight 4/3 approximation for capacitated vehicle routing in trees}.
\newblock In {\em Approximation, Randomization, and Combinatorial Optimization
  (APPROX/RANDOM)}, volume 116, pages 3:1--3:15, 2018.

\bibitem[BKS17]{becker2017quasi}
A.~Becker, P.~N. Klein, and D.~Saulpic.
\newblock A quasi-polynomial-time approximation scheme for vehicle routing on
  planar and bounded-genus graphs.
\newblock In {\em 25th Annual European Symposium on Algorithms (ESA)}, 2017.

\bibitem[BKS18]{becker2018polynomial}
A.~Becker, P.~N. Klein, and D.~Saulpic.
\newblock Polynomial-time approximation schemes for $k$-center, $k$-median, and
  capacitated vehicle routing in bounded highway dimension.
\newblock In {\em 26th Annual European Symposium on Algorithms (ESA)}, 2018.

\bibitem[BKS19]{becker2019ptas}
A.~Becker, P.~N. Klein, and A.~Schild.
\newblock {A PTAS for bounded-capacity vehicle routing in planar graphs}.
\newblock In {\em Workshop on Algorithms and Data Structures}, pages 99--111.
  Springer, 2019.

\bibitem[BP19]{becker2019framework}
A.~Becker and A.~Paul.
\newblock A framework for vehicle routing approximation schemes in trees.
\newblock In {\em Workshop on Algorithms and Data Structures}, pages 112--125.
  Springer, 2019.

\bibitem[BSvdSS14]{BSSS14}
S.~C. Boyd, R.~Sitters, S.~van~der Ster, and L.~Stougie.
\newblock The traveling salesman problem on cubic and subcubic graphs.
\newblock {\em Mathematical Programming}, 144(1-2):227--245, 2014.

\bibitem[BTV22]{blauth2021improving}
J.~Blauth, V.~Traub, and J.~Vygen.
\newblock Improving the approximation ratio for capacitated vehicle routing.
\newblock {\em Mathematical Programming (to appear)}, 2022.

\bibitem[CFKL20]{cohen2020light}
V.~{Cohen-Addad}, A.~Filtser, P.~N. Klein, and H.~Le.
\newblock On light spanners, low-treewidth embeddings and efficient traversing
  in minor-free graphs.
\newblock In {\em Symposium on Foundations of Computer Science (FOCS)}, pages
  589--600. IEEE, 2020.

\bibitem[Chr76]{christofides1976worst}
N.~Christofides.
\newblock Worst-case analysis of a new heuristic for the travelling salesman
  problem.
\newblock {\em Technical Report 388, Graduate School of Industrial
  Administration, Carnegie Mellon University}, 1976.

\bibitem[CL12]{crainic2012fleet}
T.~G. Crainic and G.~Laporte.
\newblock {\em Fleet management and logistics}.
\newblock Springer Science \& Business Media, 2012.

\bibitem[CLS15]{CLS15}
J.~R. Correa, O.~Larr{\'{e}}, and J.~A. Soto.
\newblock {TSP} tours in cubic graphs: beyond 4/3.
\newblock {\em {SIAM} J. Discret. Math.}, 29(2):915--939, 2015.

\bibitem[DKM17]{DKM17}
Z.~Dvor{\'{a}}k, D.~Kr{\'{a}}l, and B.~Mohar.
\newblock Graphic {TSP} in cubic graphs.
\newblock In {\em {STACS}}, volume~66, pages 27:1--27:13, 2017.

\bibitem[DL18]{DL18}
B.~Dun{\'{\i}}k and R.~Lukotka.
\newblock Cubic {TSP:} {a} 1.3-approximation.
\newblock {\em {SIAM} J. Discret. Math.}, 32(3):2094--2114, 2018.

\bibitem[DM15]{das2015quasipolynomial}
A.~Das and C.~Mathieu.
\newblock {A quasipolynomial time approximation scheme for Euclidean
  capacitated vehicle routing}.
\newblock {\em Algorithmica}, 73(1):115--142, 2015.

\bibitem[DR59]{dantzig1959truck}
G.~B. Dantzig and J.~H. Ramser.
\newblock The truck dispatching problem.
\newblock {\em Management Science}, 6(1):80--91, 1959.

\bibitem[FMRS22]{friggstad2021improved}
Z.~Friggstad, R.~Mousavi, M.~Rahgoshay, and M.~R. Salavatipour.
\newblock Improved Approximations for Capacitated Vehicle Routing with
  Unsplittable Client Demands.
\newblock In {\em Integer Programming and Combinatorial Optimization - 23rd
  International Conference, IPCO}, 2022.

\bibitem[Gao13]{Gao13}
Z.~Gao.
\newblock An LP-based $\frac{3}{2}$-approximation algorithm for the $s$-$t$
  path graph traveling salesman problem.
\newblock {\em Oper. Res. Lett.}, 41(6):615--617, 2013.

\bibitem[GKP95]{GKP95}
M.~Grigni, E.~Koutsoupias, and C.~H. Papadimitriou.
\newblock An approximation scheme for planar graph {TSP}.
\newblock In {\em {FOCS}}, pages 640--645. IEEE, 1995.

\bibitem[GLS05]{GLS05}
D.~Gamarnik, M.~Lewenstein, and M.~Sviridenko.
\newblock An improved upper bound for the {TSP} in cubic 3-edge-connected
  graphs.
\newblock {\em Oper. Res. Lett.}, 33(5):467--474, 2005.

\bibitem[GMZ22]{GMZ22}
F.~Grandoni, C.~Mathieu, and H.~Zhou.
\newblock Unsplittable Euclidean Capacitated Vehicle Routing: A
  $(2+\epsilon)$-Approximation Algorithm.
\newblock {\em arXiv preprint arXiv:2209.05520}, 2022.

\bibitem[GRW08]{golden2008vehicle}
B.~Golden, S.~Raghavan, and E.~Wasil.
\newblock {\em The vehicle routing problem: latest advances and new
  challenges}, volume~43 of {\em Operations Research/Computer Science
  Interfaces Series}.
\newblock Springer, 2008.

\bibitem[HR85]{haimovich1985bounds}
M.~Haimovich and A.~H.~G. {Rinnooy~Kan}.
\newblock Bounds and heuristics for capacitated routing problems.
\newblock {\em Mathematics of Operations Research}, 10(4):527--542, 1985.

\bibitem[JS22]{jayaprakash2021approximation}
A.~Jayaprakash and M.~R. Salavatipour.
\newblock Approximation schemes for capacitated vehicle routing on graphs of
  bounded treewidth, bounded doubling, or highway dimension.
\newblock In {\em ACM-SIAM Symposium on Discrete Algorithms (SODA)}, pages
  877--893, 2022.

\bibitem[KKO21]{karlin2021slightly}
A.~R. Karlin, N.~Klein, and S.~{Oveis Gharan}.
\newblock A (slightly) improved approximation algorithm for metric TSP.
\newblock In {\em Proceedings of the 53rd Annual ACM SIGACT Symposium on Theory
  of Computing (STOC)}, pages 32--45, 2021.

\bibitem[KR16]{KR16}
J.~Karp and R.~Ravi.
\newblock A $\frac{9}{7}$-approximation algorithm for graphic {TSP} in cubic
  bipartite graphs.
\newblock {\em Discret. Appl. Math.}, 209:164--216, 2016.

\bibitem[LSL90]{li1990worst}
C.~L. Li and D.~Simchi-Levi.
\newblock Worst-case analysis of heuristics for multidepot capacitated vehicle
  routing problems.
\newblock {\em ORSA Journal on Computing}, 2(1):64--73, 1990.

\bibitem[LSLD92]{li1992distance}
C.~L. Li, D.~Simchi-Levi, and M.~Desrochers.
\newblock On the distance constrained vehicle routing problem.
\newblock {\em Operations research}, 40(4):790--799, 1992.

\bibitem[Mos98]{mosheiov1998vehicle}
G.~Mosheiov.
\newblock Vehicle routing with pick-up and delivery: tour-partitioning
  heuristics.
\newblock {\em Computers \& Industrial Engineering}, 34(3):669--684, 1998.

\bibitem[MS16]{MS16}
T.~M{\"o}mke and O.~Svensson.
\newblock Removing and adding edges for the traveling salesman problem.
\newblock {\em Journal of the ACM (JACM)}, 63(1):1--28, 2016.

\bibitem[Muc14]{mucha2014frac}
M.~Mucha.
\newblock $\frac{13}{9}$-approximation for graphic {TSP}.
\newblock {\em Theory of computing systems}, 55(4):640--657, 2014.

\bibitem[MZ21]{MZ21}
C.~Mathieu and H.~Zhou.
\newblock Probabilistic analysis of {Euclidean} capacitated vehicle routing.
\newblock In {\em International Symposium on Algorithms and Computation
  (ISAAC)}, page 43:1–43:16, 2021.

\bibitem[MZ22a]{MZ22unit}
C.~Mathieu and H.~Zhou.
\newblock A {PTAS} for capacitated vehicle routing on trees.
\newblock In {\em Proceedings of the International Colloquium on Automata,
  Languages and Programming (ICALP)}, page 95:1–95:20, 2022.

\bibitem[MZ22b]{MZ22unsplit}
C.~Mathieu and H.~Zhou.
\newblock A tight $(1.5+ \epsilon)$-approximation for unsplittable capacitated
  vehicle routing on trees.
\newblock {\em arXiv preprint arXiv:2202.05691}, 2022.

\bibitem[New20]{New20}
A.~Newman.
\newblock An improved analysis of the M{\"{o}}mke-Svensson algorithm for
  graph-TSP on subquartic graphs.
\newblock {\em {SIAM} J. Discret. Math.}, 34(1):865--884, 2020.

\bibitem[OSS11]{gharan2011randomized}
S.~{Oveis Gharan}, A.~Saberi, and M.~Singh.
\newblock A randomized rounding approach to the traveling salesman problem.
\newblock In {\em Symposium on Foundations of Computer Science (FOCS)}, pages
  550--559. IEEE, 2011.

\bibitem[Ser78]{Ser78}
A.~I. Serdyukov.
\newblock On some extremal walks in graphs (in Russian).
\newblock {\em Upravlyaemye Sistemy}, 17:76--79, 1978.

\bibitem[SV14]{sebHo2014shorter}
A.~Seb{\H{o}} and J.~Vygen.
\newblock Shorter tours by nicer ears: 7/5-approximation for the graph-TSP, 3/2
  for the path version, and 4/3 for two-edge-connected subgraphs.
\newblock {\em Combinatorica}, pages 1--34, 2014.

\bibitem[TV02]{toth2002vehicle}
P.~Toth and D.~Vigo.
\newblock {\em The Vehicle Routing Problem}.
\newblock SIAM, 2002.

\bibitem[TV18]{TV18}
V.~Traub and J.~Vygen.
\newblock Beating the integrality ratio for $s$-$t$-tours in graphs.
\newblock In {\em {FOCS}}, pages 766--777. IEEE, 2018.

\bibitem[TVZ20]{TVZ20}
V.~Traub, J.~Vygen, and R.~Zenklusen.
\newblock Reducing path {TSP} to {TSP}.
\newblock In {\em {STOC}}, pages 14--27. {ACM}, 2020.

\bibitem[vBS20]{BVS20}
R.~van Bevern and V.~A. Slugina.
\newblock A historical note on the 3/2-approximation algorithm for the metric
  traveling salesman problem.
\newblock {\em Historia Mathematica}, 53:118--127, 2020.

\bibitem[Vyg12]{Vyg12}
J.~Vygen.
\newblock New approximation algorithms for the TSP.
\newblock {\em OPTIMA}, 90:1–12, 2012.

\bibitem[vZ18]{Zuy18}
A.~van Zuylen.
\newblock Improved approximations for cubic bipartite and cubic {TSP}.
\newblock {\em Math. Program.}, 172(1-2):399--413, 2018.

\end{thebibliography}

\end{document}